\documentclass{article}

\usepackage[english]{babel}
\usepackage{todonotes}
\usepackage{tikz}
\usetikzlibrary{patterns}
\usepackage{authblk}
\usepackage{pgfplots}
 \usepackage{fullpage}%for printing

\usepackage{amsthm}
\usepackage{thmtools}
\usepackage{amsmath,amssymb}
\usepackage{mathtools}
\usepackage{graphicx}
\usepackage[colorlinks=true, allcolors=blue]{hyperref}

\newtheorem{definition}{Definition}[section]
\newtheorem{theorem}[definition]{Theorem}
\newtheorem{proposition}[definition]{Proposition}

\newtheorem{corollary}[definition]{Corollary}

\title{Transforming the Challenge of Constructing Low-Discrepancy Point Sets into a Permutation Selection Problem}

\author[1]{Fran\c{c}ois Cl\'ement \thanks{francois.clement@lip6.fr}}
\author[1]{Carola Doerr \thanks{carola.doerr@lip6.fr}}
\author[2]{Kathrin Klamroth \thanks{klamroth@math.uni-wuppertal.de}}
  \author[3]{Lu\'is Paquete \thanks{paquete@dei.uc.pt}}
  \affil[1]{Sorbonne Universit\'e, CNRS, LIP6, Paris, France}
  \affil[2]{University of Wuppertal, School of Mathematics and Natural Sciences, IMACM, Gaußstr.~20, 42119 Wuppertal, Germany}
  \affil[3]{CISUC, Department of Informatics Engineering, University of Coimbra, Portugal}

\begin{document}
\maketitle

\begin{abstract}
Low discrepancy point sets have been widely used as a tool to approximate continuous objects by discrete ones in numerical processes, for example in numerical integration. Following a century of research on the topic, it is still unclear how low the discrepancy of point sets can go; in other words, how regularly distributed can points be in a given space. Recent insights using optimization~\cite{CDKP} and machine learning~\cite{MPMC} techniques have led to substantial improvements in the construction of low-discrepancy point sets, resulting in configurations of much lower discrepancy values than previously known. Building on the optimal constructions from~\cite{CDKP}, we present a simple way to obtain $L_{\infty}$-optimized placement of points that follow the same relative order as an (arbitrary) input set. 
Applying this approach to point sets in dimensions 2 and 3 for up to 400 and 50 points, respectively, we obtain point sets whose $L_{\infty}$ star discrepancies are up to 25\% smaller than those of the current-best sets from~\cite{MPMC}, and around 50\% better than classical constructions such as the Fibonacci set.
\end{abstract}

\section{Introduction}
Discrepancy measures quantify how uniformly a point set is distributed in a given space. There exist a number of different discrepancy measures, among which one of the most important is the $L_{\infty}$ star discrepancy. Its importance stems from the Koksma-Hlawka inequality~\cite{Koksma,Hlawka}, bounding the error made when approximating an integral by the average of a finite number of local evaluations. The $L_{\infty}$ star discrepancy of a finite point set $P$ in the unit cube $[0,1]^d$, $d\geq 1$, corresponds to the worst absolute difference between the proportion of points that fall inside a box $[0,q)\subset[0,1]^d$ and the volume of that box. More formally, it is given by
\begin{equation}\label{def:discr}
d^{*}_{\infty}(P):=\sup_{q \in [0,1)^d}\bigg| \frac{|P\cap [0,q)|}{|P|}-\lambda([0,q))\bigg|,    
\end{equation}
where $\lambda$ is the Lebesgue measure. Points with small $L_{\infty}$ star discrepancy, \emph{low-discrepancy point sets or sequences}, have been extensively used in a wide variety of contexts~\cite{CauwetCDLRRTTU20,GalFin,MatBuilder,SantnerDoE}, of which the most important is numerical integration \cite{DickP10}.

\begin{figure}[t]
    \centering
    \includegraphics[width=0.6\textwidth]{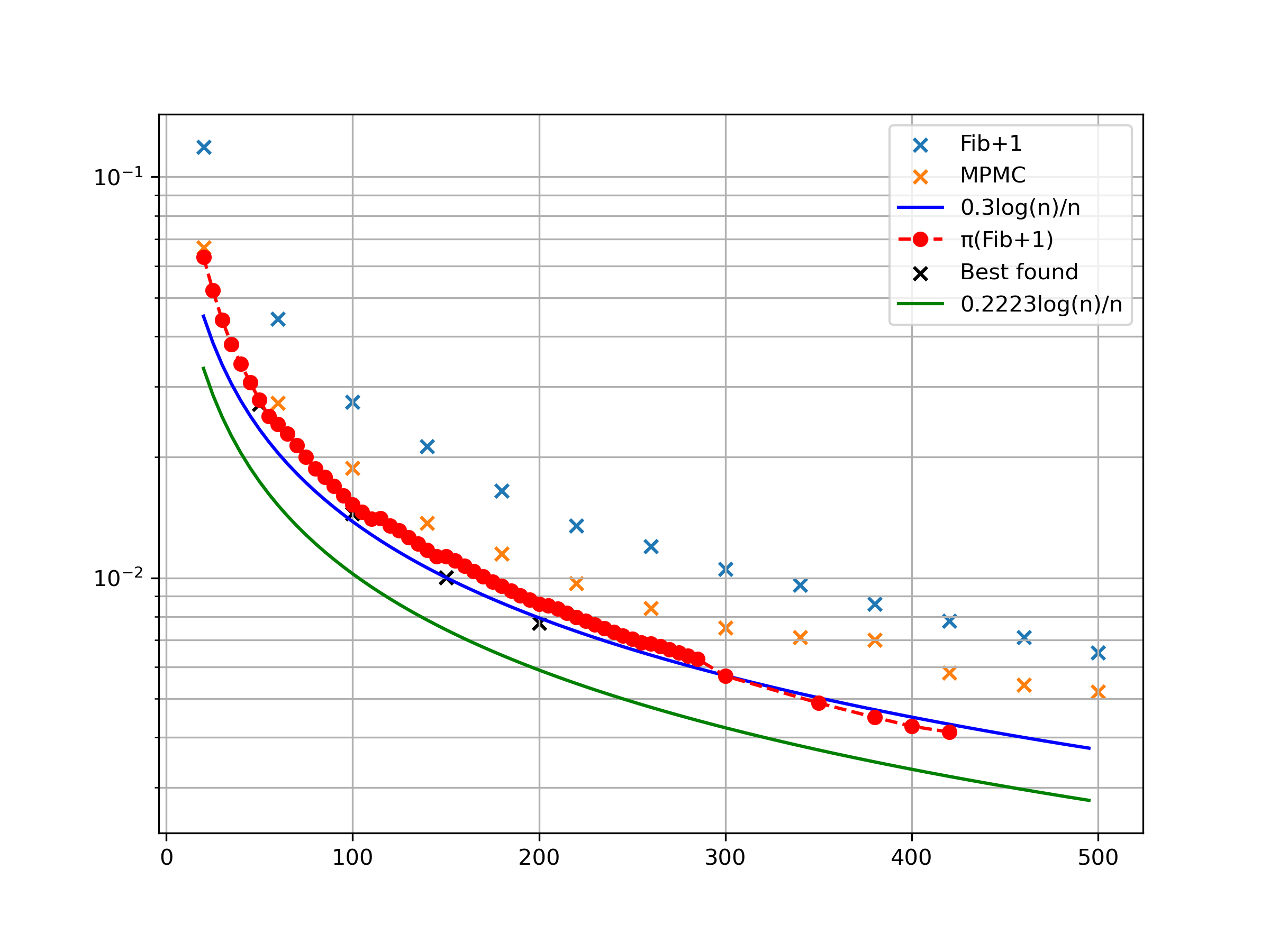}
    \caption{Best discrepancy values obtained with our permutation approach $\pi$(Fib+1), compared to the traditional Fibonacci set shifted by one point (Fib+1, see Sections~\ref{sec:classic} and~\ref{sec:shift}) and the MPMC results from~\cite{MPMC}. The black crosses labelled ``Best found'' correspond to the best values obtained via a more extensive search of the different shifts at the start of the Fibonacci sequence (see Figure~\ref{fig:Shifts}). The current best theoretical upper bound for the asymptotic discrepancy order, obtained by Ostromoukhov~\cite{Ostro} for a \emph{far} larger number of points (between $60^7$ and $60^8$), is given by the green line $0.2223\log(n)/n$. The $0.3\log(n)/n$ curve is added as a comparison point.}
    \label{fig:BEST}
\end{figure}

Despite extensive research on this topic (see for example the books~\cite{DickP10,Nie92,Pano,Mat,NW} to name a few), a number of key questions remain, among them the best possible asymptotic discrepancy order. In dimension $d=2$, the optimal asymptotic discrepancy order is known to be $\Theta(\log(n)/n)$ as shown in~\cite{Schmidt}, where $n$ is the number of points, but the optimal constants are not known. For the analogous one-dimensional sequences (see~\cite{Roth54} for the relationship), the lower bound is 0.065 by Larcher and Puchhammer in~\cite{LarPuch}, while the upper bound 0.22223 was obtained by Ostromoukhov with a permutation of the base 60 van der Corput sequence~\cite{Ostro}. On the other hand, no asymptotically tight bounds are available in higher dimensions. To this day, the best lower bound is by Bilyk, Lacey, and Vagharshakyan~\cite{BilykSmall} in $\Omega(\log(n)^{d/2+c(d)}/n)$, where
$c(d)$ is a constant depending only on the dimension.
The best currently-known constructions have a discrepancy of $O(\log(n)^{d-1}/n)$.

Setting aside the asymptotic regime, even less is known for bounds for the best possible discrepancy value $d^*_{\infty}(n,d)$ for fixed values of $n$ and $d$. It is known that $d^*_{\infty}(n,d) \leq c\sqrt{d/n}$ where $c$ is a constant~\cite{Heinrich}, but the proof is not constructive and, despite a variety of approaches tackling the construction of low-discrepancy point sets~\cite{DoerrCompobyCompo,DoerrGneSriv,GneBrack2007}, no explicit constructions achieving this bound are known. Lower bounds in this context are also rare. It was shown by White~\cite{whit:onop:1976} that $d^*_{\infty}(n,2)\geq 1/n$ for $n\geq 4$, extended in~\cite{CDKP} to $d^*_{\infty}(n,d)\geq 1/n$ for $n\geq 3$ and $d\geq 3$. Optimal constructions are known only in very specific settings, either for low dimensions~\cite{whit:onop:1976,CDKP} or for one or two points~\cite{PTV,LarPuch}.

We introduced two different mathematical programming models to build provably optimal point sets in~\cite{CDKP}. The second of these, \emph{the assignment formulation}, defines a point set in dimension 2 by considering two lists of coordinates and an assignment matrix, linking the list of coordinates to the points' positions. For an $n$-point set, choosing an assignment matrix can also be seen as choosing a permutation of $\{1,\ldots,n\}$. In this paper, we build on this approach by removing the optimality constraint and focusing our work on choosing appropriate permutations. Indeed, choosing the assignment is the most difficult part of the model and fixing it thus greatly lessens the computational burden. We are therefore able to reach a much higher number of points with excellent discrepancy values than in our previous models in~\cite{CDKP}. In particular, our results in dimension 2 shown in Figure~\ref{fig:BEST} outperform the very recent and currently best-performing machine learning approach in~\cite{MPMC}. Our approach can also be seen as transforming the problem of constructing low-discrepancy point sets, with $n$ points in $[0,1]^d$, to that of selecting $d-1$ adequate permutations.

\section{Optimizing Point Sets for a Given Permutation}

\subsection{The Global Assignment Formulation}
For the sake of completeness, we recall the assignment model introduced in~\cite{CDKP} for the construction of optimal point sets in dimension $2$. It consists of two sets of sorted variables $x$ and $y$ representing respectively the set of first and second coordinates of the points, whose placement will be optimized using the discrete structure of the $L_{\infty}$ star discrepancy computation shown in~\cite{Nie92}. This is then combined with a permutation matrix $a$ to obtain the exact point set, i.e., $a$ identifies the assignment of $x$- and $y$-coordinates in the constructed point set. The continuous variable $f$ corresponds to the $L_\infty$ star discrepancy value of the point set defined by $x,y$ and $a$, and is to be minimized.

\begin{subequations}\label{eq:M5_2dcont}
\begin{align}
    \min\;\; & f && \nonumber\\ 
    \text{s.t.}\;\; & \displaystyle \frac{1}{n} \sum_{u=1}^i\sum_{v=1}^j a_{uv} - x_{i}y_{j} \leq f && \forall i,j=1,\dots,n \label{eq:M5_upperboundcont}\\
    & \displaystyle \frac{-1}{n} \sum_{u=1}^{i-1}\sum_{v=1}^{j-1} a_{uv} + x_{i}y_{j} \leq f && \forall i,j=1,\dots,n+1 \label{eq:M5_lowerboundcont}\\
    & x_{n+1}=1,\; y_{n+1}=1 && \label{eq:M5_dummy}\\
    & x_{i+1}-x_{i} \geq \varepsilon && \forall i=1,\dots,n-1 \label{eq:M5_4f}\\
    & y_{i+1}-y_{i} \geq \varepsilon && \forall i=1,\dots,n-1 \label{eq:M5_4fbis}\\
    & \sum_{i=1}^n a_{ij} = 1 && \forall j=1,\dots,n \label{eq:M5_column}\\
    & \sum_{j=1}^n a_{ij} = 1 && \forall i=1,\dots,n \label{eq:M5_row}\\
    & x_i,y_i\in[0,1] && \forall i=1,\dots,n\; \nonumber \\
    & a_{ij}\in\{0,1\} && \forall i,j=1,\dots,n\; \nonumber \\
    & f\geq 0 .\hspace*{-4cm} &&\nonumber
\end{align}
\end{subequations}

 We point out the \emph{general position} hypothesis in constraints~\eqref{eq:M5_4f} and~\eqref{eq:M5_4fbis}, specifying that no two points share a coordinate. We showed in~\cite{CDKP} that there is an optimal point set verifying this hypothesis in two dimensions. In our previous approach, no supplementary information was given to the solver apart from extra constraints to limit symmetries and variable range. It would then solve the model to optimality.

In our new approach, the permutation matrix $a$ is fixed before optimization. The solver then only needs to find the optimal values of $x$ and $y$, i.e., of the coordinates, given the permutation of the points. This allows for much faster solving times, even for instances of a few hundred points. In particular, excellent, if not optimal, solutions are found within seconds. Since we are no longer trying to find a provably optimal set, we will always use the general position hypothesis, even in higher dimensions.

\subsection{\texorpdfstring{$L_{\infty}$}{L-infty} Star Point Sets from Random Permutations and Classical Sequences}\label{sec:classic}

Clearly, when pre-specifying the assignment matrix $a$ in model \eqref{eq:M5_2dcont} and optimizing only for the coordinates $x$ and $y$, we obtain an upper bound $f=f(a)$ on the optimal $L_\infty$ star discrepancy value the quality of which depends on the choice of $a$. In two dimensions, this assignment matrix can be identified with a permutation $\pi$.
The challenge of identifying point sets of low discrepancy is now transformed into choosing the permutation $\pi$ well, i.e., such that the obtained upper bound is as close as possible -- or potentially equal -- to the optimal $L_\infty$ star discrepancy value. Unfortunately, exhaustive search over all $n!$ permutations is clearly impossible for even a moderate number of points. Natural candidates are the permutations defined by classical low-discrepancy point sets, such as the Fibonacci set, the van der Corput sequence, or the Kritzinger sequence~\cite{Kritz} ``lifted'' to two dimensions. Since it will be useful for us in the later sections, we recall that the Fibonacci set is the two-dimensional set version of the Kronecker sequence with golden ratio. It is defined as $\{(i/n,\{i*\phi\}): i \in \{1,\ldots,n\}\}$, where $\phi$ is the golden ratio and $\{x\}$ is the fractional part of $x$. We will call \emph{lifting} the process of transforming a $d-1$ dimensional sequence to a $d$ dimensional set, by setting the $d$-th coordinate of the $i$-th point to $i/n$. This is a common practice in discrepancy, already used by Roth in~\cite{Roth54}, and does not change the discrepancy order of the resulting set.

A natural first step is to try random permutations of $\{1,\ldots,n\}$, with results here for $n=100$. Figure~\ref{fig:Rando} shows these permutations do not perform so well. Indeed, the majority of values are between 0.03 and 0.04, while already the unoptimized shifted by one Fibonacci set for 100 points has a discrepancy of 0.0261. 

\begin{figure}
    \centering
    \includegraphics[width=0.5\textwidth]{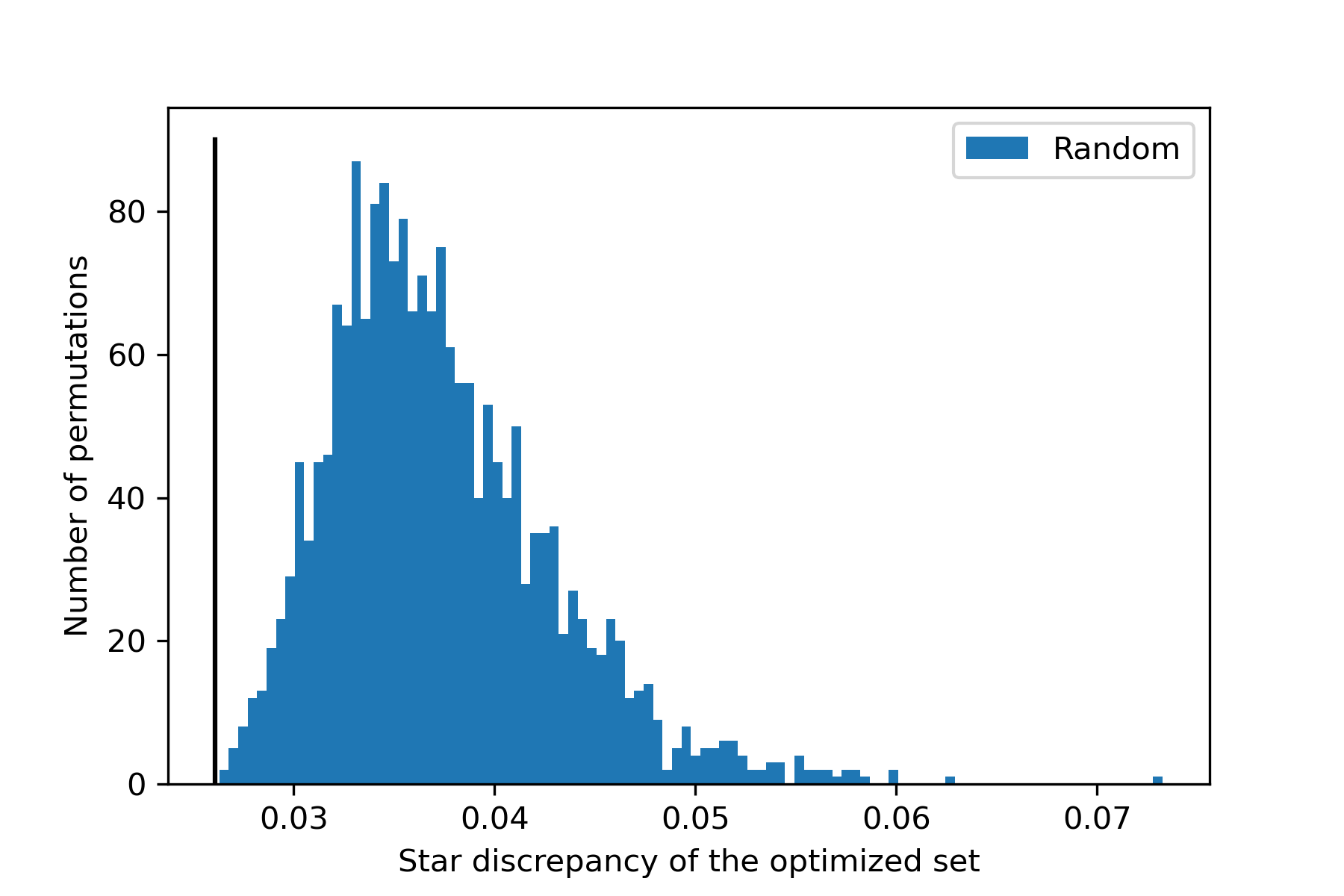}
\caption{Optimized discrepancy values $f(\pi)$ obtained by taking 2000 random permutations, $n=100$. We recall that the discrepancy for the (unoptimized but shifted by 1) Fibonacci set at 100 points is 0.0261, represented by the black vertical line in the plot.}
    \label{fig:Rando}
\end{figure}

These results for random permutations at $n=100$ can be compared to those obtained with permutations from usual constructions in Table~\ref{tab:usu}. While the common constructions do not perform equally well, the permutations generated by the van der Corput sequence and the Fibonacci set lead to excellent results, largely outperforming those from MPMC. Furthermore, this trend seems to be increasing as $n$ goes up.

\begin{table}
    \caption{Optimized discrepancy values $f(\pi)$ for the Kritzinger, van der Corput and Fibonacci sets in dimension 2. The MPMC values from~\cite{MPMC} are added as reference, but their associated permutations were not entered in model \eqref{eq:M5_2dcont} to further optimize the coordinates. For further reference, the best values we obtained during all our tests were 0.02709 for 50 points (as good as previous constructions for 100 points), 0.014444 for 100 points, 0.01 for 150 points, and 0.007719 for 200 points, just below the MPMC value for $n=300$.}
    \centering
    {\footnotesize{
    \begin{tabular}{|c|c|c|c|c|c|c|c|c|c|c|c|}
    \hline
    $n$ & 20& 50 &100 &180 &220&260& 350 & 420 & 500\\
    \hline
    
    MPMC &0.0666& -  &0.0188 &0.0115& 0.0097& 0.0084& -&0.0058& 0.0052\\
    \hline
    $\pi$(van der Corput) & 0.06562&0.02899&  0.01588 & 0.00932& 0.00784& 0.00689& - &- &-\\
    \hline
    $\pi$(Fibonacci) & 0.06319&0.02733& 0.01525&  0.00954& 0.00799 & 0.006862& - &- &- \\
    \hline
    $\pi$(Fibonacci+1)& 0.06219&0.02742&  0.01492& 0.00901& 0.00737 & 0.00640 & 0.004872& 0.00412& -\\
    \hline 
    Fibonacci+1&0.105883 &0.049165 & 0.026105 &0.015165 &0.012407 & 0.010894& 0.008731 &0.00728 &0.00611 \\
    \hline

    \end{tabular}
    }}
    \label{tab:usu}
\end{table}

For all these instances, while proving optimality of $f(\pi)$ when solving \eqref{eq:M5_2dcont} with given $\pi$ is a more difficult question, nearly optimal sets are found within minutes if not seconds. In particular, we note most of the time taken by the solver is in the presolving and setting up the model. Once solving begins in earnest, good candidates are found very quickly. For all tests in this paper, each problem was given a 60 or 600 seconds (depending on $n$ and $d$) time limit to find the best solution, which was returned afterwards. The only exceptions are the largest problems ($n>300$), which were given an hour time limit. Each problem was run on a single machine (whether from the MeSu cluster at Sorbonne Université or a laptop). As such, this approach can be seen as relatively inexpensive from a computational perspective.

Results suggest this method is able to find excellent point sets very quickly, up to around 400 points as was shown in Figure~\ref{fig:BEST}. However, the poorer results with random permutations imply that the permutation should be chosen with care. The assignment associated with the Fibonacci set gives an excellent point set, with discrepancy 0.01525, much better than that of the initial Fibonacci set 0.027485 and the roughly 0.019 bound obtained with our previous models and a 40\,000s runtime limit. Shifting the Fibonacci set (and hence modifying the associated permutation $\pi$) by one point, starting at $(0,\{\phi\})$ rather than $(0,0)$, gives even better results: 0.01492 against 0.0261 for the initial set.

\subsection{\texorpdfstring{$L_{\infty}$}{L-infty} Star Point Sets from Shifted Sequences}\label{sec:shift}

While classical sequences like the Fibonacci sequence usually start with a point at $(0,0)$, it turns out that shifting this first point to a later value of the sequence leads to permutations $\pi$ that perform significantly better when used as the basis in model \eqref{eq:M5_2dcont}. Indeed, we have no requirement to keep the first point in the lower-left corner. While it was shown to play an important role for $(t,m,s)$-nets in~\cite{OwenSobol}, we have no such constraint for a Kronecker-like construction. It turns out that by starting the sequence at the second point, $(0,\{\phi\})$ and extracting the corresponding permutation $\pi$, one obtains already very different results for $f(\pi)$. We performed experiments on the Fibonacci sequence, but this could be generalized to any set generated by a lifted sequence. Figure~\ref{fig:Shifts} shows how the discrepancy values $f(\pi)$ evolve as we change the starting point of the Fibonacci sequence for $n=$50, 100, 150, and 200.

\begin{figure}
    \centering
    \includegraphics[width=0.24\textwidth]{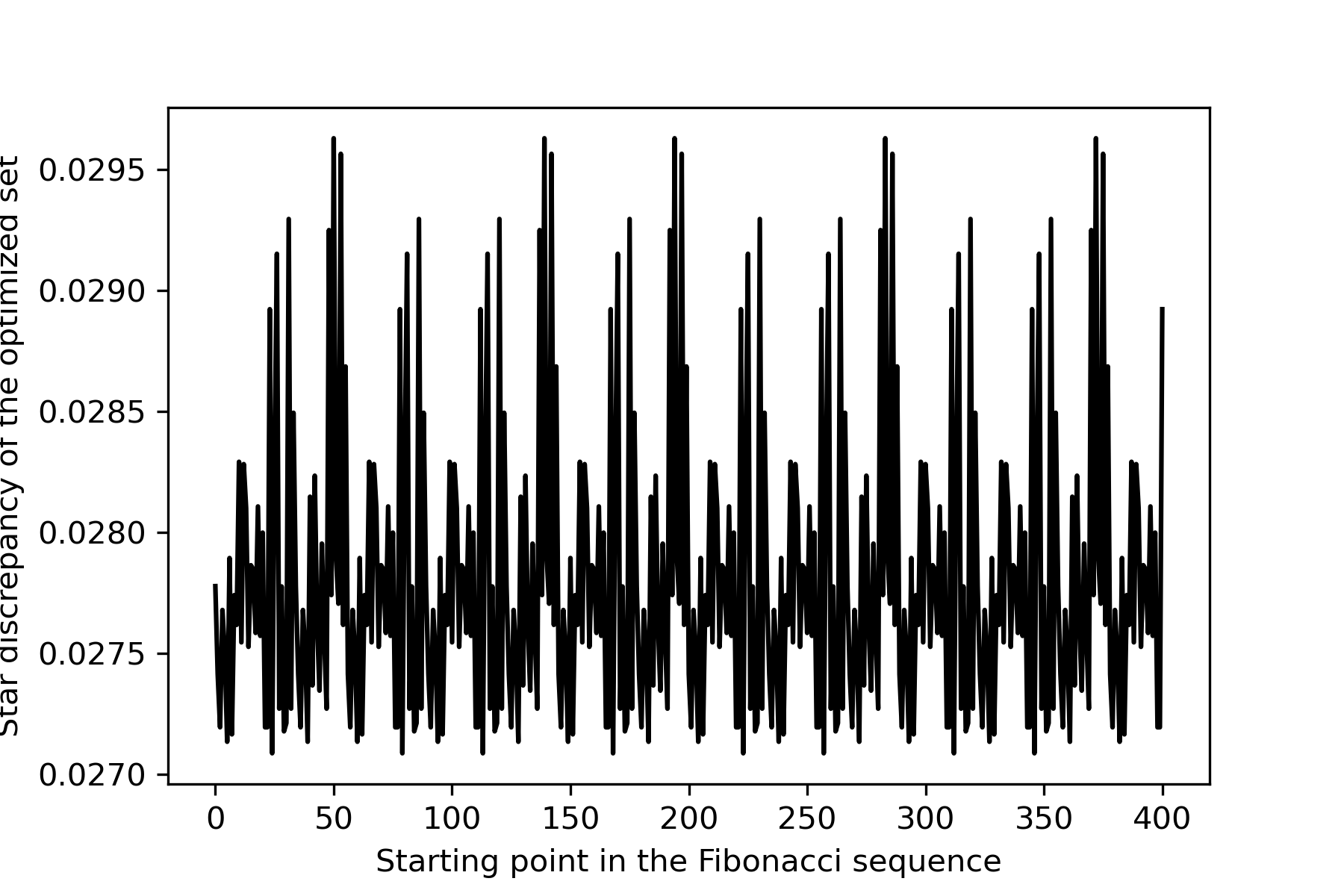}
    \includegraphics[width=0.24\textwidth]{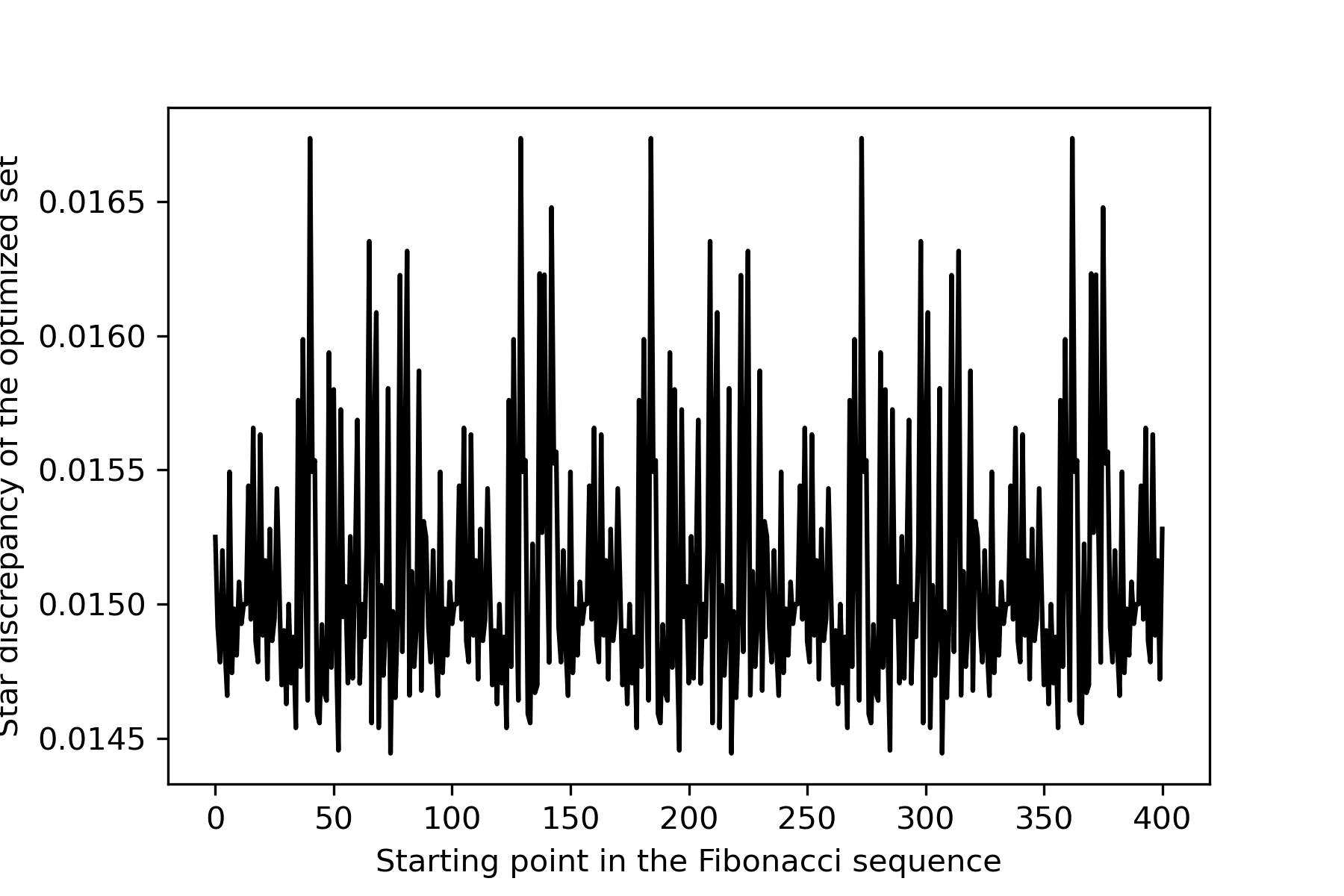}
    \includegraphics[width=0.24\textwidth]{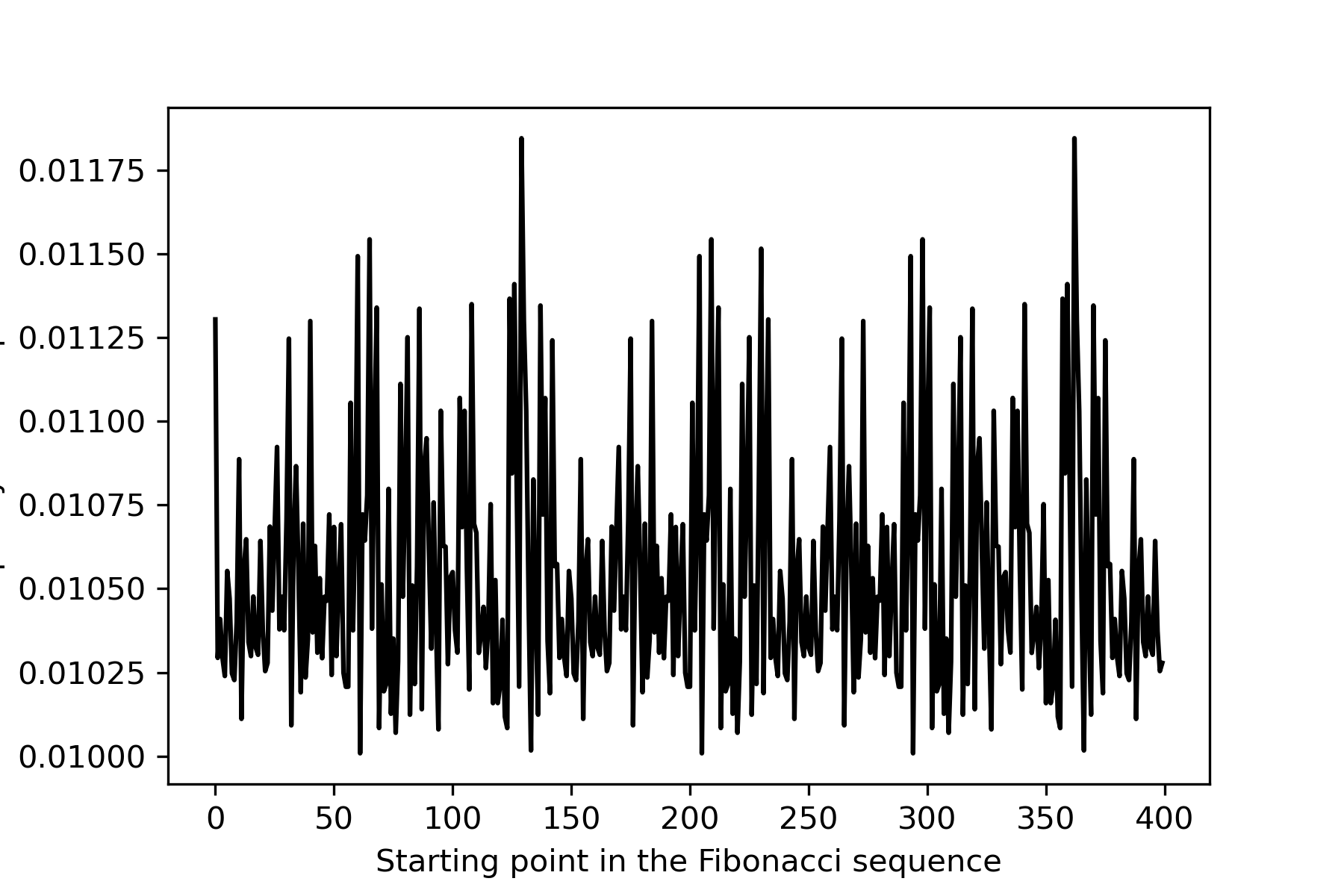}
    \includegraphics[width=0.24\textwidth]{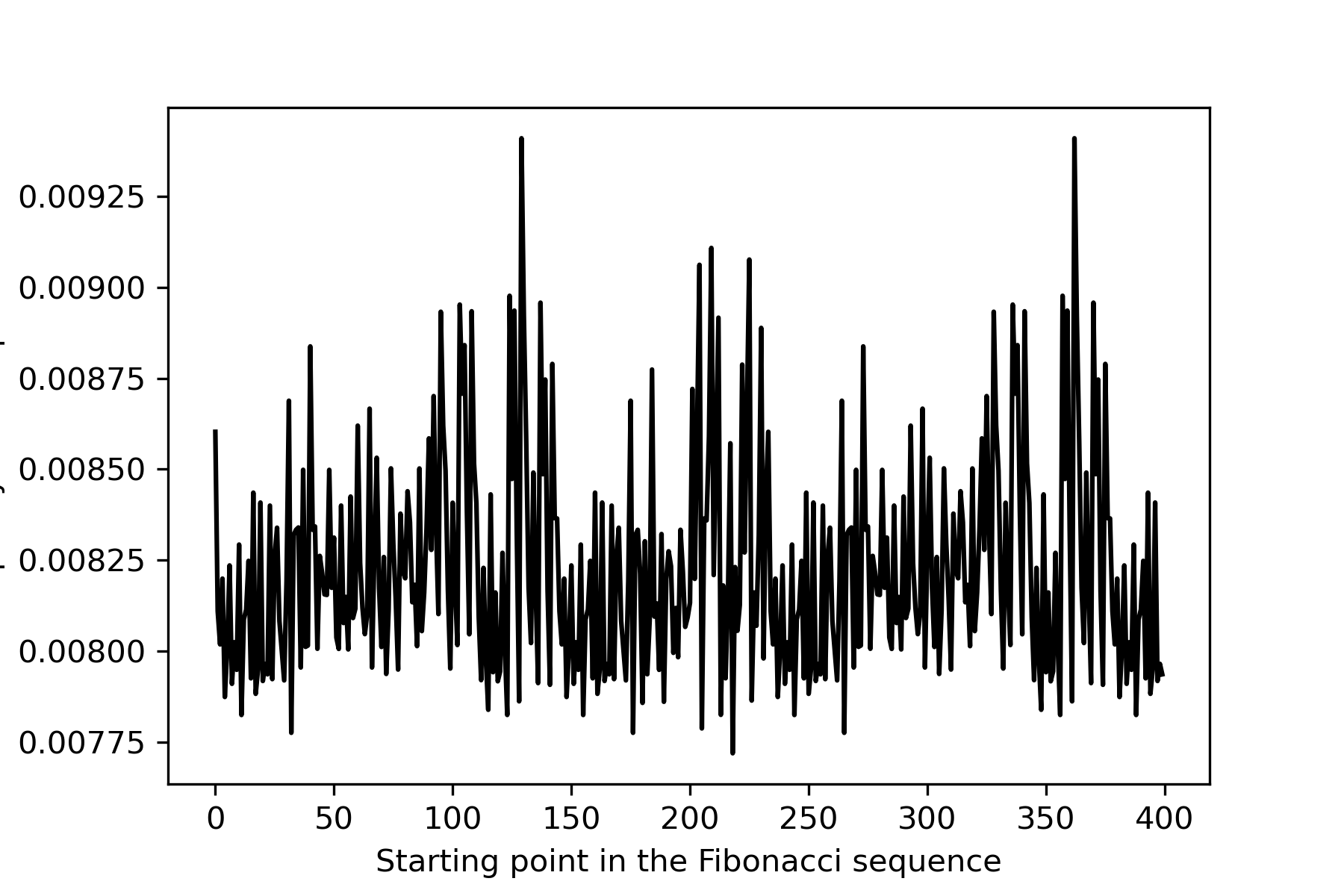}
    \caption{Evolution of the optimized $L_\infty$ star discrepancy $f(\pi)$ obtained with model \eqref{eq:M5_2dcont} for permutations $\pi$ obtained from shifted Fibonacci sequences, depending on the starting point of the sequence,  
    for $n=$50, 100, 150, 200, from left to right. The best values obtained are 0.027088 for $n=50$, 0.014444 for $n=100$, 0.01 for $n=150$ and 0.007718 for $n=200$.}
    \label{fig:Shifts}
\end{figure}

Unfortunately, from these four simple examples, we do not observe either periodicity nor a regularly good starting point for the Fibonacci sequence. Nevertheless, it appears that already starting in $(0,\{\phi\})$ rather than $(0,0)$ gives noticeable improvements in all the plots, while never performing badly. Figure~\ref{fig:5shifts} shows the discrepancy values obtained by taking $(0,0)$, $(0,\{\phi\})$ and 3 random starting indices $j$ leading to starting points in $(0,\{j*\phi\})$ (the values of $j$ are changed in each instance). While the differences are minor, starting in $(0,\{\phi\})$ seems to be a consistently better than average choice. In particular, it is better in nearly all cases than starting in $(0,0)$. As such, we use this starting point in our experiments over all $n$ presented in Figure~\ref{fig:BEST}. We point out that changing this starting point also improves the Fibonacci set itself, without optimization.

\begin{figure}
    \centering
    \includegraphics[width=0.6\textwidth]{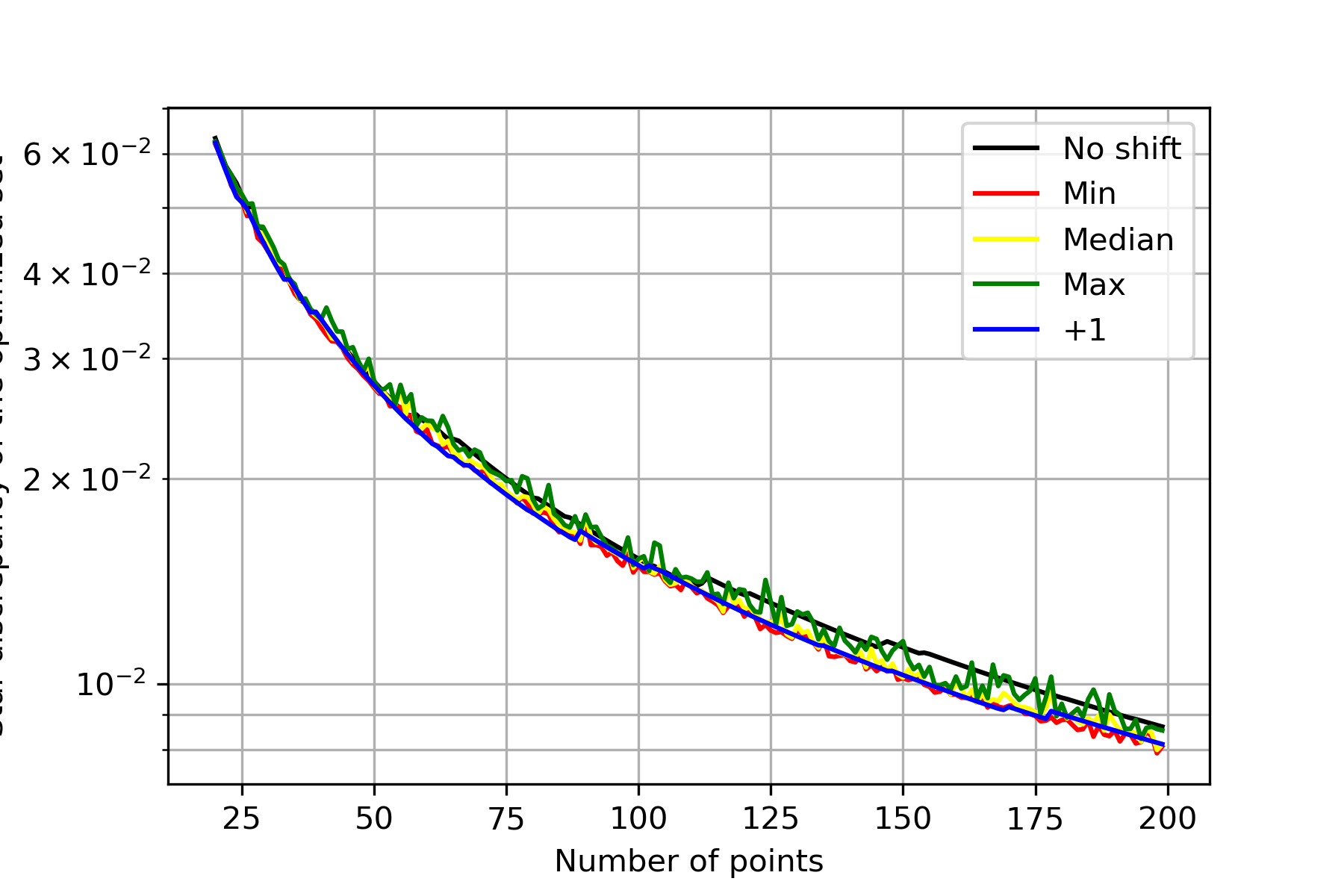}
    \caption{Discrepancies obtained by optimizing the Fibonacci permutation starting in 0 (No shift), 1 (+1), or three random integers (which we then split into a minimum, median and maximum value curves). Starting in 1 seems to be a reliably good choice, consistently outperforming the start in 0 and generally also better than the random choices.}
    \label{fig:5shifts}
\end{figure}

Finally, we provide a set of images to represent our best point set obtained with the permutation $\pi$ from a shifted Fibonacci set for 100 points (starting at $(0,\{74\phi\})$), a plot of the points and different discrepancy heatmaps: all the local discrepancies, only closed boxes and only open boxes. In this context, closed boxes correspond to boxes with too many points inside, while open boxes correspond to boxes with too few, and can be associated to constraints~\eqref{eq:M5_lowerboundcont} and~\eqref{eq:M5_upperboundcont}. We also provide an illustration of the difference between the initial Fibonacci set $F_{74+100}$ and the one obtained after optimization $F^*$. We note that while the points in the last figure look very close to each other, it is mostly because there is an underlying grid of 100 by 100 where each points can move in approximately a single box. One key observation is that closed boxes seem to be less of a limiting factor than for the original Fibonacci set (see heatmaps in~\cite{CDKP}). While there is still some room for improvement, the overall balance seems much better, with both closed and open boxes reaching the worst discrepancy values.

\begin{figure}
    \centering
    \includegraphics[width=0.325\textwidth]{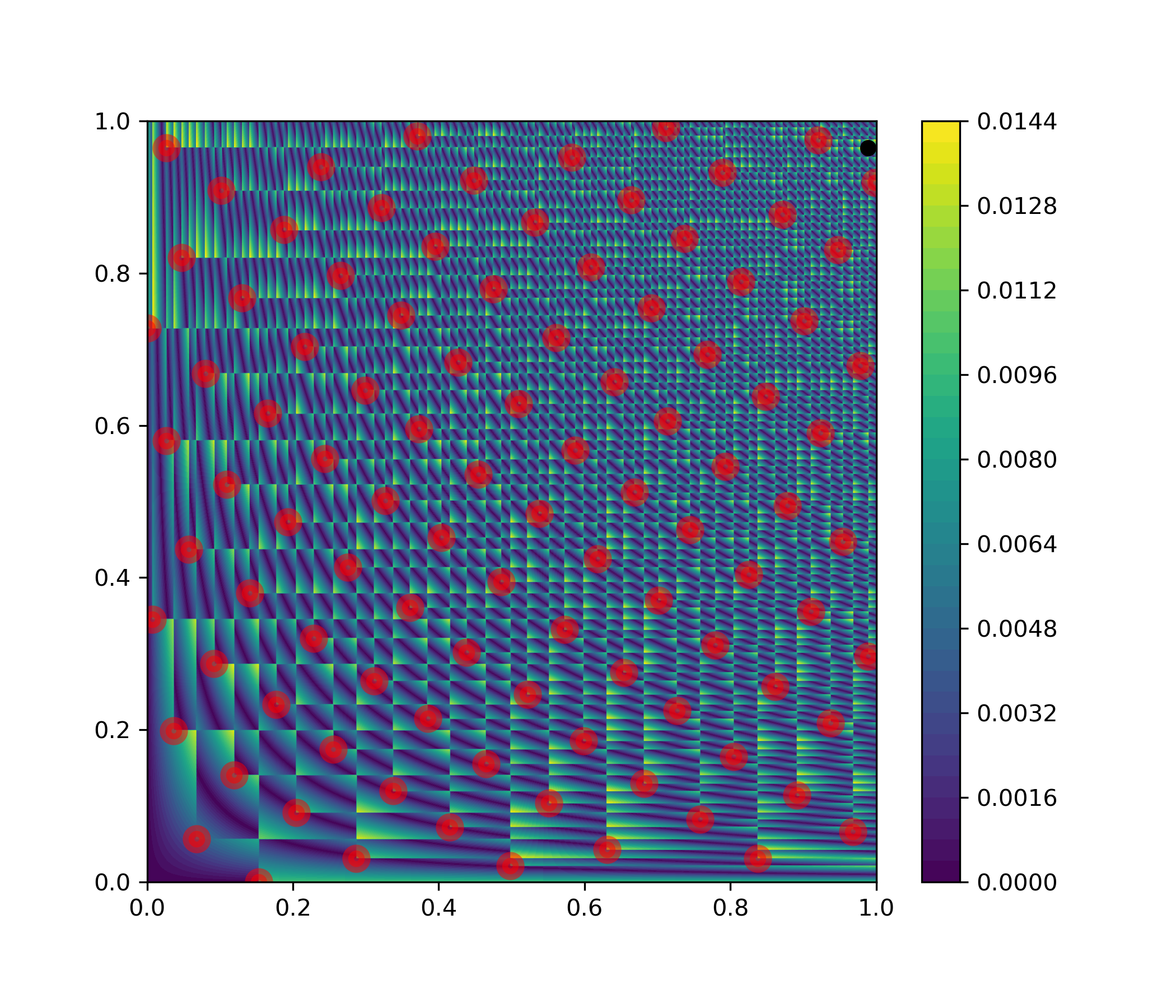}
    \includegraphics[width=0.325\textwidth]{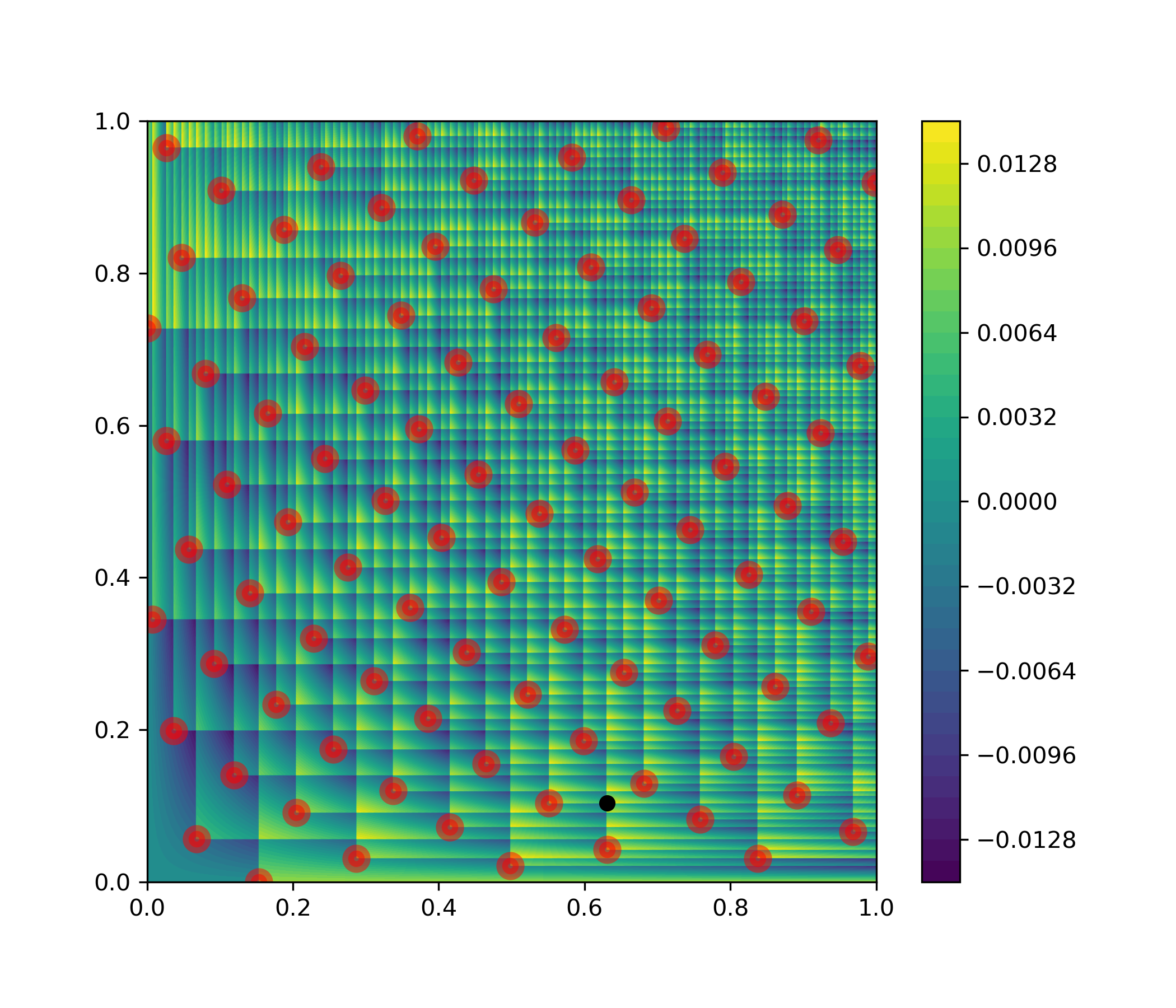}
    \includegraphics[width=0.325\textwidth]{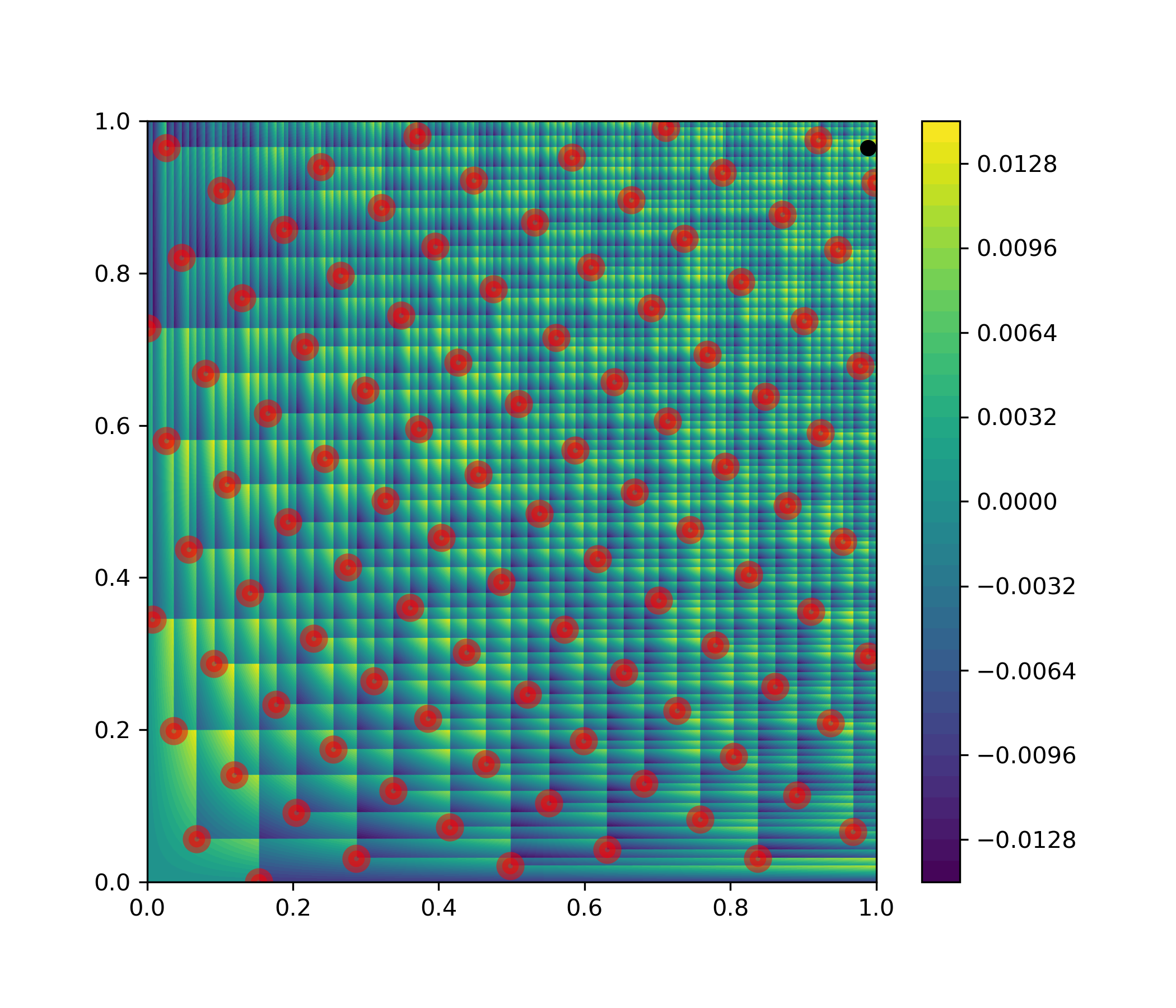}
    \includegraphics[width=0.325\textwidth]{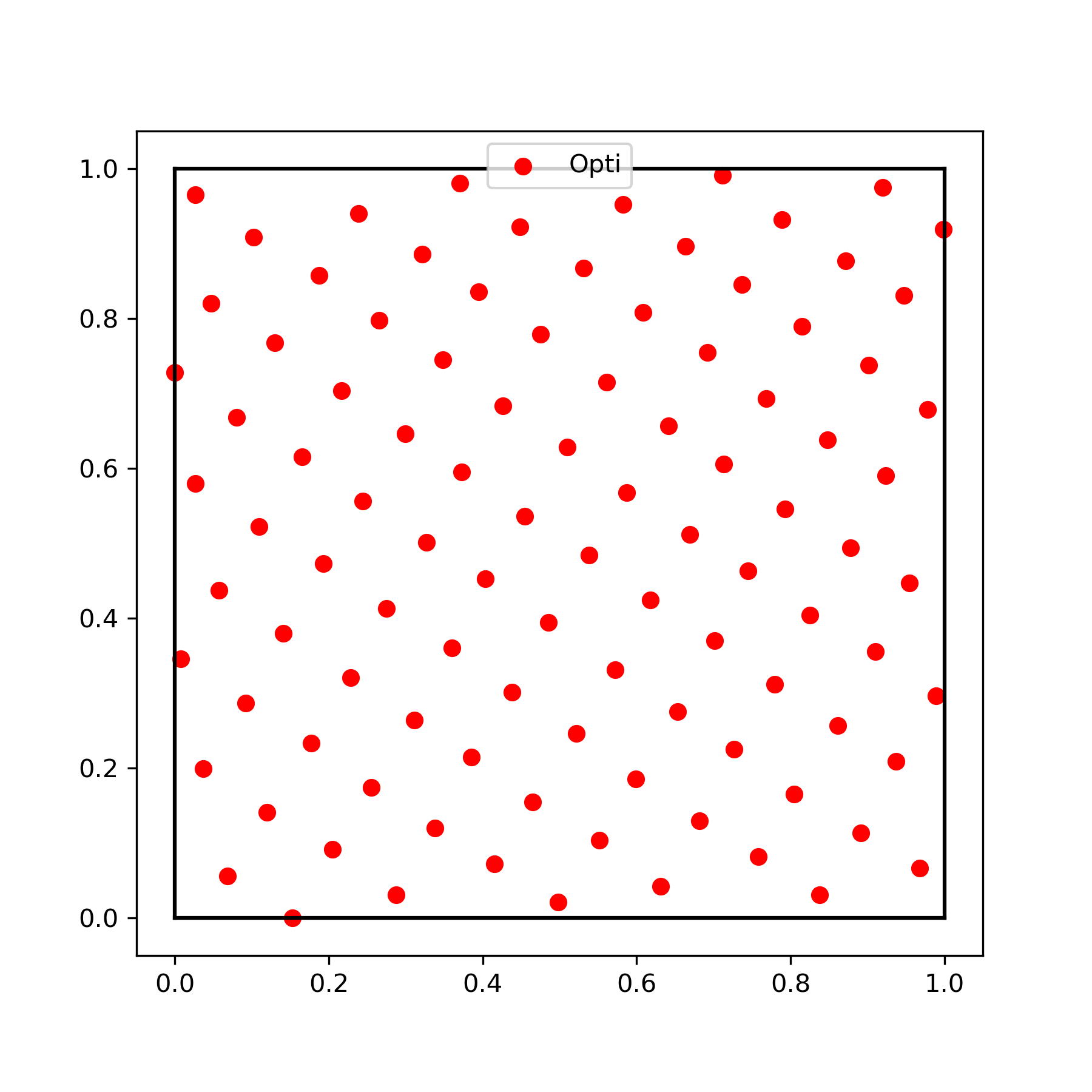}
    \includegraphics[width=0.325\textwidth]{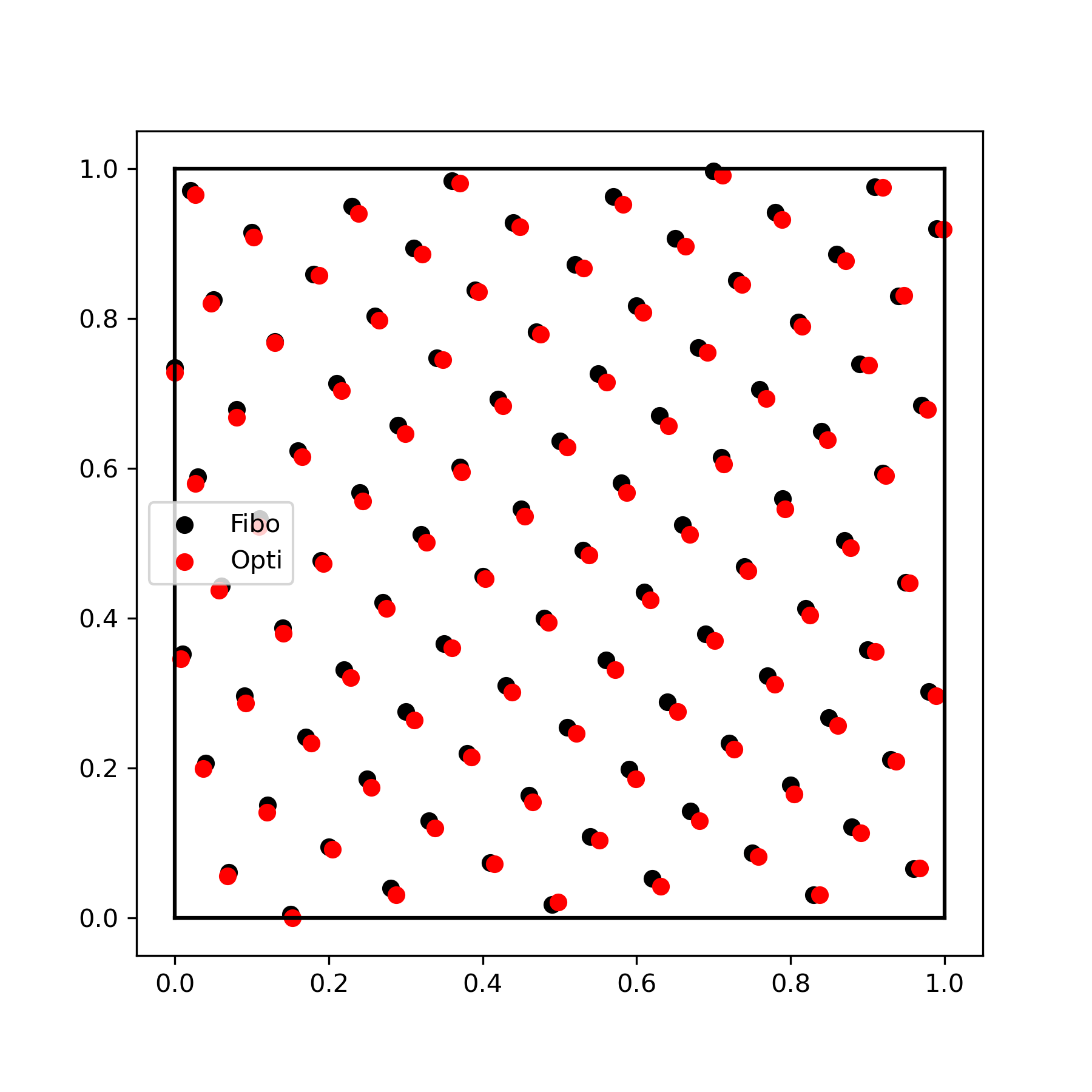}
    \caption{Point sets obtained by taking the permutation for the Fibonacci set starting at the 75th point of the Fibonacci sequence. From first to last:  local discrepancies of $F^*$, local discrepancies of open boxes for $F^*$, local discrepancies of closed boxes for $F^*$, the new optimized 100 point set $F^*$, and a comparison between the offset between the Fibonacci set $F_{74+100}$ (black) and $F^*$ (red).}
    \label{fig:enter-label}
\end{figure}

\subsection{\texorpdfstring{$L_{\infty}$}{L-infty} Star Point Sets in Three Dimensions}

As was shown in~\cite{CDKP}, the two-dimensional assignment model can be extended to three dimensions. The only issue is that the quadratic volume terms become a product of three variables in three dimensions. Fixing the permutation can therefore be done in this context as well, also leading to excellent results past the 8 point limit for optimal sets from \cite{CDKP}, as Table~\ref{tab:3d} shows. These results align with our observations in two dimensions. However, we now lack the excellent Fibonacci set to compare to. Comparisons are therefore done only with the Sobol' and Kritzinger sequences and a point set communicated to us by the authors of~\cite{MPMC}. All computations for $n\leq 38$ were done with a 10 minute runtime, and the $n\geq 42$ with a one hour time limit.

\begin{table}[h]
    \centering
    \caption{Discrepancy values obtained in three dimensions from the permutations generated by the lifted two-dimensional Kritzinger and Sobol' sequences (with an (L)), and the three-dimensional Sobol' sequence. We add the non-optimized Sobol' sequence to compare with. The lifted sequences perform better than the sequences of a dimension higher. For reference, the best 50 point set obtained with MPMC had a discrepancy of 0.05534.}
    \vspace{0.5pt}
    {\small{
    \begin{tabular}{|c|c|c|c|c|c|c|c|c|c|}
    \hline
         $n$&14& 18& 22&26&30&34&38&42& 50  \\
         \hline
         $\pi$(Kritzinger (L))& 0.1480&0.1222& 0.1044&0.09091&0.08265&0.07213&0.06545&0.05839& 0.05226\\
         \hline
         $\pi$(Sobol' (L))&0.14286&0.1136& 0.1002&0.08882&0.08016&0.07635&0.06947&0.06836&0.05500\\
         \hline
         $\pi$(Sobol')& 0.1429& 0.125&0.1103&0.09977&0.0885&0.07983&0.07183&0.07073&0.05978\\
         \hline
         Sobol'& 0.23096& 0.21788& 0.17602&0.16657& 0.13322& 0.17243 & 0.13068& 0.14273& 0.08997\\
         \hline
    \end{tabular}
    }}
    \label{tab:3d}
\end{table}

\section{Investigating the Structure of Good Permutations}
Following the results from the previous sections, the natural question is to try to find if there are simple ways of characterizing whether a permutation will lead to good optimized sets or not. While we have not found permutations generally better than that induced by the Fibonacci sequence (with shifts), we have no proof for the optimality of shifted Fibonacci sequences in this respect. %of this result. 
The van der Corput sequence, which performed very well as was shown in Table~\ref{tab:usu}, has a very structured permutation $\pi$: the difference between $\pi(i)$ and $\pi(i+1)$ can take only very specific values depending on $\lfloor \log_2(n)\rfloor$ and $n$. $\lfloor \log_2(n)\rfloor$ defines the set of possible values, while $n-2^{\lfloor \log_2(n)\rfloor}$ dictates where the offsets of 1 happens for these values. While relatively easy to describe, this does not seem easy to generalize to a larger set of good candidate permutations.

\subsection{Kronecker Permutations}

Rank-one lattices whose first point is $(0,0)$ are relatively few, around $O(n^2)$, and can be exhaustively searched. Indeed, such a lattice $L(r)$ is defined in the following way $L(r):=\{(i/n,\{ir\}:i \in \{1,\ldots,n\} \}$, where $\{ir\}$ is the fractional part of $ir$. We will call the permutations resulting from these lattices \emph{Kronecker permutations} as they are the two-dimensional set version of Kronecker sequences in one-dimension. The Fibonacci set is one such permutation. In case of equality between multiple coordinates $i<j<k$, we set $\pi(i)<\pi(j)<\pi(k)$.

\begin{theorem}
    Given an initial point in $(0,0)$, there exist only $O(n^2)$ different Kronecker permutations generated by the associated lattices.
\end{theorem}

\begin{proof}
    As $r$ goes from 0 to 1, one obtains all the different rank-one lattices possible. Since each point moves continuously over the two-dimensional torus, there can be a change from one permutation to the next if and only if two points are aligned horizontally. By periodicity of the lattice construction, two points can be aligned horizontally if and only if the origin point $(0,0)$ is aligned with another point. We now need to count how many times such an event can happen for $r\in [0,1]$. Since the second coordinate is defined by $\{ir\}$, these events correspond exactly to pairs $(i,r)$ such that $ir$ is an integer. Given that $i$ is itself a positive integer, bounded by $n$, this can happen only when $r$ is a rational whose denominator is at most $n$. There are $\Theta(n^2)$ such rationals in $[0,1]$ and thus $O(n^2)$ different Kronecker permutations (some may be identical).
\end{proof}

This first result greatly reduces the permutation search space from $n!$ to $\Theta(n^2)$. While it is possible that whole classes of good permutations are lost by considering only Kronecker permutations, Figure~\ref{fig:all_perms} suggests that Kronecker permutations provide point sets of excellent quality, better than all previously known constructions, and of far lower $L_{\infty}$ star discrepancy than randomly chosen permutations. Despite far more outliers, associated to permutations such as the identity or the decreasing permutation, the best Kronecker permutations clearly outperform the best random ones. Indeed, the best value found is at 0.0149, against 0.0263 for random permutations. In particular, a much larger proportion of permutations have lower discrepancy than the best random permutations, while vastly reducing the overall search space.

\begin{figure}
    \centering
    \includegraphics[width=0.45\textwidth]{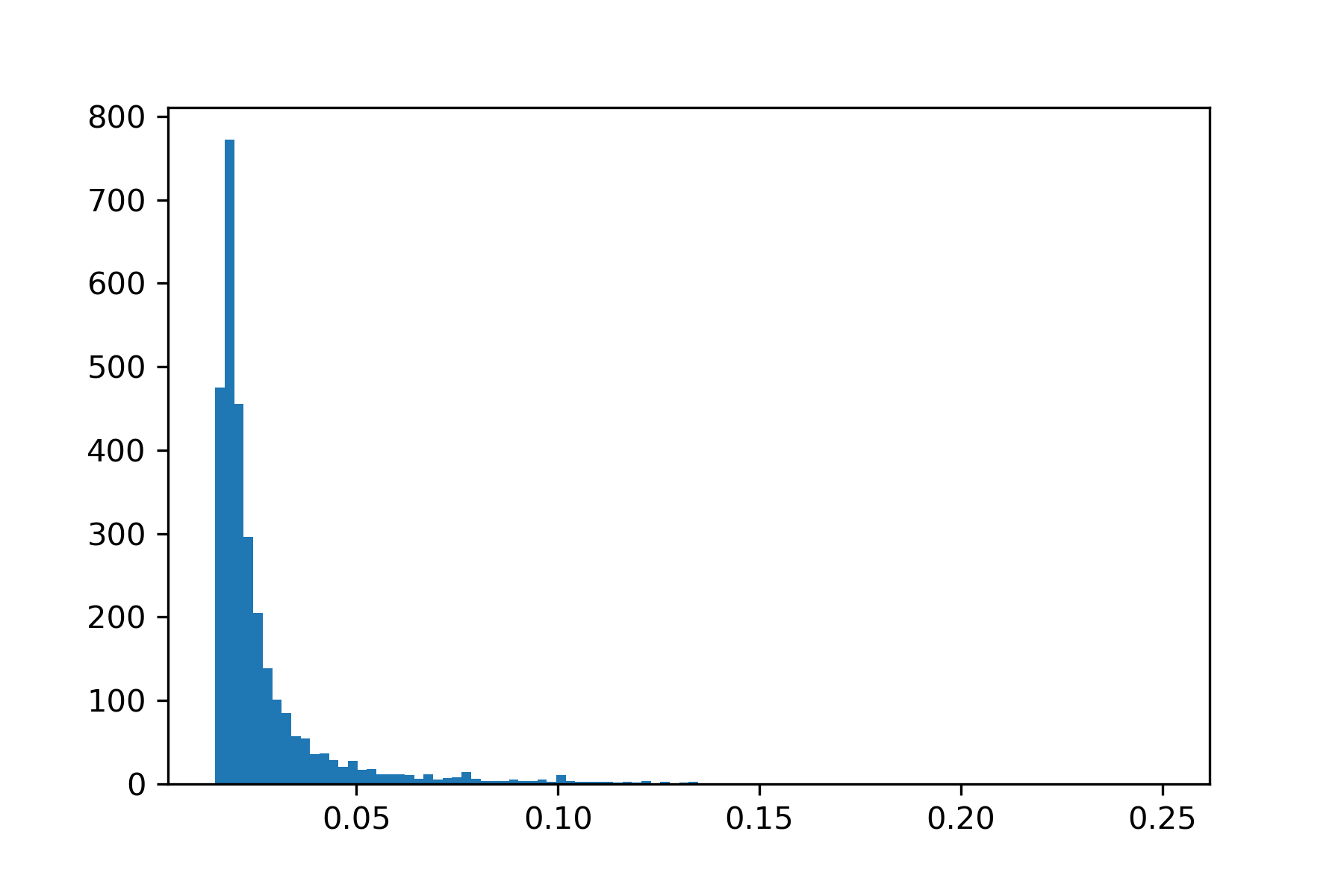}
    \includegraphics[width=0.45\textwidth]{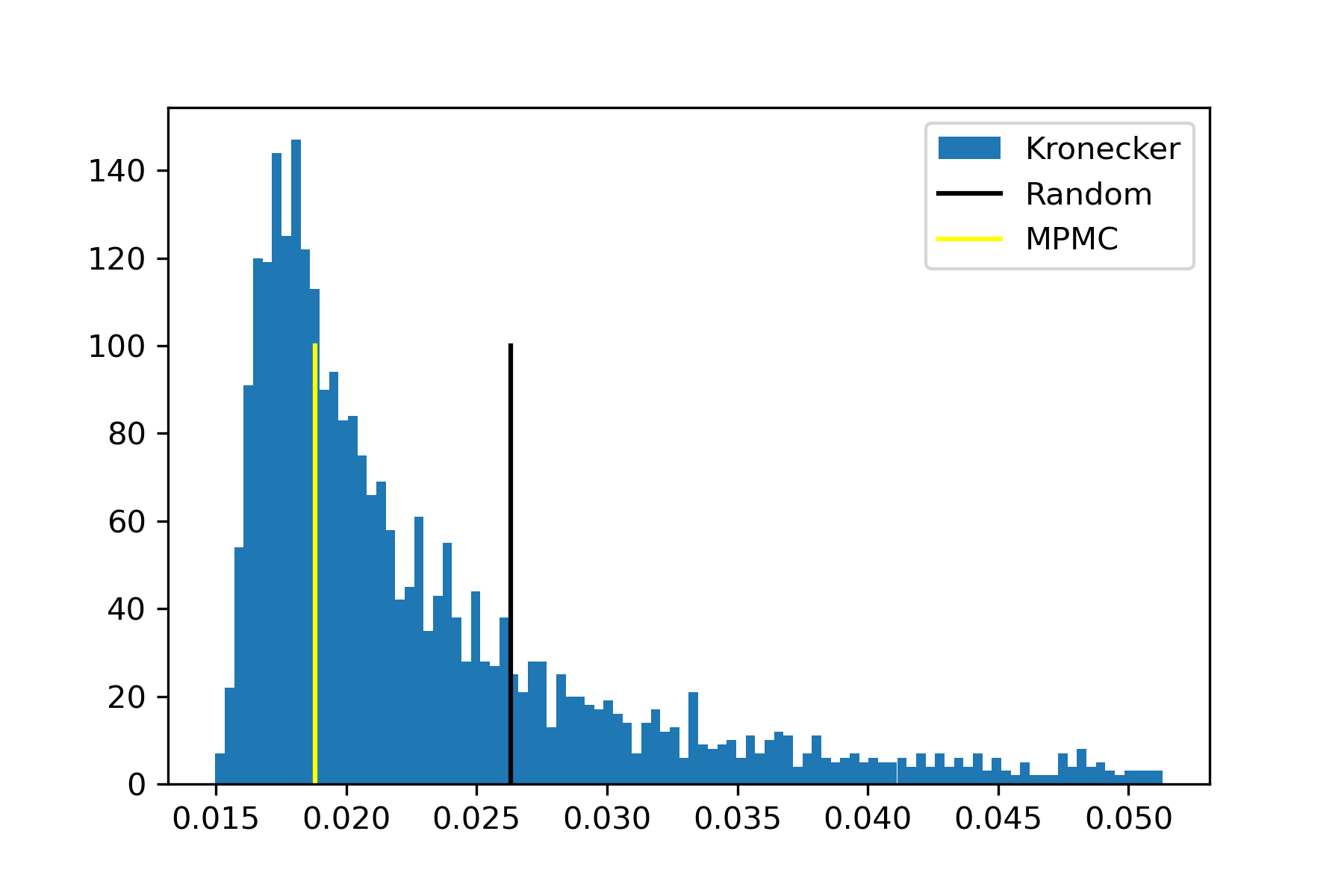}
    \caption{Histograms plotting the discrepancy values $f(\pi)$ obtained from all Kronecker permutations for 100 points, after optimization. The histogram on the left is complete, while the one on the right is focused on the densest areas from the left histogram. For reference, the (unoptimized but shifted by one) Fibonacci set of 100 points has a discrepancy of 0.0261, the best MPMC result from~\cite{MPMC} is 0.0188 (yellow vertical line), and our best result from the timeout table in~\cite{CDKP} was 0.0193 at 40\,000s. We also add a black vertical line to represent the best result from all random permutations.}
    \label{fig:all_perms}
\end{figure}

Though not directly relevant for this paper, we note that the theorem can be extended to higher dimensions, where the lattice is given by $\{(i/n,\{ir_1\},\ldots,\{ir_{d-1}\}):i \in \{1,\ldots,n\}\}$, with $d-1$ lattice parameters~$r_i$.

\begin{corollary}
    There exist $O(n^{2d-2})$ Kronecker permutations in dimension $d$.
\end{corollary}
This result is obtained by observing that, by the previous theorem, each coordinate barring the first has $O(n^2)$ possible permutations. Coming back to dimension 2, we can further characterize Kronecker permutations with the following proposition.

\begin{proposition}[\cite{Sos}]
Let $\pi$ be a Kronecker permutation of $\{1,\ldots,n\}$. Then $\pi(i+1)-\pi(i)$ can take only three values: $\pi(2)$, $\pi(2)+1$, or $n-\pi(2)$.
\end{proposition}
This result is used to show the three-gap theorem by Sós in~\cite{Sos}, which states that given $n$ points of a Kronecker sequence on the one-dimensional torus, there are at most three distinct distances between two consecutive points on the torus.

From the permutation point of view, a Kronecker permutation can therefore be characterized by $\pi(2)$, the relative position of the first non-zero point, and $\pi(n)$, the relative position of the last point. The first determines the allowed $\pi(j+1)-\pi(j)$ values, while the second determines where the $\pi(2)+1$ changes will take place. From our current experiments, it does not look like the value of $\pi(2)$ is easy to choose. Indeed, Figure~\ref{fig:EvolLatt} shows that good parameters seem to be relatively evenly spread over the interval $[1/n,1]$. We note that a very similar plot can be obtained by plotting the discrepancy values for the \emph{non-optimized} lattices with the same parameters, suggesting that good low-discrepancy lattices will generally translate to good Kronecker permutations for optimization.

\begin{figure}
    \centering
    \includegraphics[width=0.6\textwidth]{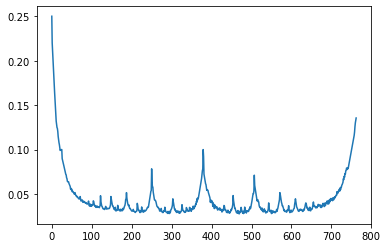}
    \caption{Evolution of the discrepancy values $f(\pi)$ obtained for $n=50$ when changing the Kronecker permutation from the identity to the nearly-decreasing permutation ($\pi(1)=1, \pi(2)=n,\ldots,\pi(n)=2$). The numbering is based on the ordering of the relevant generating rationals in $[1/n,1]$.}
    \label{fig:EvolLatt}
\end{figure}

\section{Conclusion}

Using nonlinear optimization models previously introduced in~\cite{CDKP}, we provide a method of obtaining excellent low-discrepancy point sets for a much higher number of points by fixing the permutation induced by the point set. This also has the advantage of changing the object we are working with from a point set in $[0,1)^d$ to $d-1$ permutations of $\{1,\ldots,n\}$. Possibly, these permutations could prove easier to work with for mathematicians when trying to provide formal bounds for the discrepancy of the sets. From a machine learning perspective, the simpler structure could also lead to new models, for example originating from computer vision.

Overall, our results show once again that the discrepancy of point sets can be reduced much further compared to traditional constructions, in particular in dimension 2. Combined with the results from~\cite{MPMC}, this suggests that much progress can still be made for the design of low-discrepancy sets. As illustrated in Figure~\ref{fig:BEST}, we anticipate that constructions in the near future to get very close to, or even outperform, theoretically proven upper bounds on the $L_{\infty}$ star discrepancy.

\section*{Acknowledgments}
This work was granted access to the HPC resources of the SACADO MeSU platform at Sorbonne Université.

This project was partially supported by PHC Procope 50969UF project funded by the German Academic Exchange Service (DAAD, Project-ID 57705092), the French Ministry of Foreign Affairs, and the French Ministry of Higher Education and Research.
This work was also partially supported by the “PHC PESSOA” program (project DISCREPANCY -- Discrepancy Problems - Algorithms and Complexity, number 49173PH), funded by the French Ministry for Europe and Foreign Affairs, the French Ministry for Higher Education and Research, and the Portuguese Foundation for Science and Technology (FCT). This work is partially funded by the FCT, I.P./MCTES through national funds (PIDDAC), within the scope of CISUC R\&D Unit - UID/CEC/00326/2020.

\bibliographystyle{alpha}
\bibliography{refs}

\end{document}